\documentclass[12pt]{article}
\usepackage{latexsym}
\usepackage{amsfonts}
\usepackage{amsmath}
\usepackage{amsmath,amssymb,proof1,latexsym}
\usepackage[all]{xypic}

\makeindex

\newcommand{\forces}{\Vdash\!}
\newcommand{\imp}{\!\rightarrow\!}
\newcommand{\imps}{\!\rightarrow\!}

\newcommand{\proves}{\vdash}

\newcommand{\lc}[1]{#1\!\!:\!\!}

\newtheorem{Prop}{\bf Proposition}

\newtheorem{Theor}{\bf Theorem}

\newtheorem{Lemma}{\bf Lemma}

\newtheorem{Coro}{\bf Corollary}

\newtheorem{Fact}{\bf Fact.}

\newtheorem{Remark}{\bf Remark}

\newtheorem{Claim}[enumi]{Claim}

\newtheorem{defin}{\bf Definition}
\newenvironment{definition}{\begin{defin} \em}{\end{defin}}
\newtheorem{exam}{\bf Example}
\newenvironment{example}{\begin{exam} \em}{\end{exam}}
\newtheorem{notat}{\bf Notation.}

\newenvironment{proof}{{\bf Proof.}}{\hfill $\slot$}
\newcommand{\slot}{\hfill \mbox{$\Box$}\vspace{\parskip}\\}
\newtheorem{Comment}{\bf Comment}

\begin{document}

\title{Epistemic Modeling with Justifications
}

\author{Sergei Artemov\\ \\
 {\small The Graduate Center, the City University of New York}\\
{\small  365 Fifth Avenue, New York City, NY 10016}\\
{\small {\tt sartemov@gc.cuny.edu}} }
\date{\today}
\maketitle

\begin{abstract}  
Existing logical models do not fairly represent epistemic situations with fallible justifications, e.g., Russell's Prime Minister example, though such scenarios have long been at the center of epistemic studies. 
We introduce \emph{justification epistemic models}, \emph{JEM}, which can handle such scenarios.
\emph{JEM} makes justifications prime objects and draws a distinction between accepted and knowledge-producing justifications; belief and knowledge become derived notions. 
Furthermore, Kripke models can be viewed as special cases of \emph{JEM}s with additional assumptions of evidence insensitivity and common knowledge of the model. 
We argue that \emph{JEM} can be applied to a range of epistemic scenarios in CS, AI, Game Theory, etc. 

\end{abstract}

\medskip\par
\section{Context and motivations}
We make the case that formal epistemology should venture, when needed, beyond modal logic limits. 
This includes a formal treatment of justifications, which have been at the heart of epistemic discussions since Plato. 

Proof systems of justification logic and general purpose classes of models for these systems have been studied in \cite{Art01a,Art08,Art12,AF15,Fit05,Fit14,KS12,Mkr97,Sed13} and many other sources. However, for formalizing epistemic scenarios, one needs specific \emph{domain-dependent models}, with additional features that are not necessary for standard soundness and completeness analysis of proof systems.  

We introduce \emph{justification epistemic models}, {\it JEM}, in which justifications are primary objects and a distinction is made between \emph{accepted} and \emph{knowledge-producing} justifications. In \emph{JEM}s, belief and knowledge are derived notions which depend on the status of supporting justifications.  We argue that {\it JEM}s can work in situations in which standard non-hyperintensional tools (Kripke, topological, algebraic) fail to fairly represent the corresponding epistemic structure. 

This paper is self-contained due to streamlining the exposition of justification models, an essential ingredient of \emph{JEM}s.

\subsection{Comparisons to other semantics of justifications}

Justification logics are obtained from the usual propositional logic by adding new atoms of the sort \[ \lc{t}F \] in which $t$ denotes a justification, $F$ a formula, and $\lc{t}F$ is read informally as 
\[ \mbox{\it t is a justification of F.}
\]

The first mathematical semantics of justification logic appeared in \cite{Art95TR} (cf. also \cite{Art01a}) for the Logic of Proofs {\sf LP}, namely, a constructive version of the {\bf canonical model} and {\bf arithmetical semantics} in which $\lc{t}F$ was interpreted as \emph{t is a proof of F}. 

In \cite{Mkr97}, the canonical model construction for {\sf LP} was formulated in a set-theoretic format, with interpreting proof objects $t$ as sets of formulas $t^\ast$ and speaking of ``$\lc{t}F$ is true" as $F\in t^\ast$. This led to generic {\bf basic models} for justification logics (cf. \cite{Art12,KS12} and the current paper). In a basic model $\ast$, 
\[ 
\models_\ast \lc{t}F\ \ \ \mbox{\it iff}\ \ \ F\in t^{\ast} .
\]
To represent multiple possible worlds, basic models were generalized to {\bf modular models} for justification logic {\sf J} in \cite{Art12}, and then extended to other justification logics in \cite{KS12}. A modular model is a Kripke model $(W,R,\models)$ with a basic model $\ast(u)$ at each world $u\in W$. Evaluation of justification assertions $\lc{t}F$ is defined canonically:
\[
u\models \lc{t}F\ \ \ \mbox{\it iff}\ \ \ F\in t^{\ast(u)}.
\]

{\bf Mkrtychev models} were also developed in \cite{Mkr97} to capture reflexivity in the Logic of Proofs. In Mkrtychev models, propositional and justification evaluations are similar to those in basic models (propositions are Boolean, justification terms are sets of formulas), but the justification assertions $\lc{t}F$ are evaluated as ``$t$ is a justification for $F$ and $F$ is true":
\[
\models_\ast \lc{t}F\ \ \ \ \ \ \mbox{\it iff}\ \ \ \ \ \ F\in t^\ast\ \mbox{\it and}\ \models_\ast F .
\]
Conceptually, Mkrtychev models are close to reflexive basic models: each reflexive basic model is a Mkrtychev model, each Mkrtychev model contains a reflexive basic model, 
However, Mkrtychev models are often easier to work with for justification logics with reflection. 

{\bf Fitting models} for {\sf LP} were introduced and developed in \cite{Fit05} and then extended to other justification logics. Conceptually, the basic model approach reflects only one reason for \emph{not knowing $F$}:  no sufficient justification for $F$ is available. In contrast, Fitting models take into account two reasons for \emph{not knowing $F$}: the \emph{Kripkean reason}, $F$ fails in some possible world; and the \emph{awareness reason}, no justification for $F$ is available. Technically, a Fitting model is a Kripke model $(W,R,\models)$ with a basic model $\ast(u)$ at each world $u\in W$ (together with appropriate closure conditions). The difference between Fitting models and modular models lies in the evaluation of justification assertions: for Fitting models, 
\[
u\models \lc{t}F\ \ \ \ \ \ \mbox{\it iff}\ \ \ \ \ \ F\in t^{\ast(u)}\ \mbox{\it and}\ \ R(u)\models F .
\]
This secures the \emph{Justification Yields Belief (JYB) property}: $\lc{t}F$ yields `$F$ is believed' (i.e., $F$ holds at all accessible worlds). 

Modular/basic models have clear motivations and provide an answer to the ontological question \emph{What kind of logical object is a justification?} However, these models might be difficult to work with, and more convenient constructions, such as Mkrtychev and Fitting models, are normally used instead. 

Technically, basic models and Mkrtychev models may be regarded as special cases of Fitting models. On the other hand, Fitting models can be identified as modular models with \textit{JYB}, cf. \cite{Art12}. This provides a natural hierarchy of the aforementioned classes of models: 
\[ \mbox{\it basic and Mkrtychev models}\ \ \subset\ \ \mbox{\it Fitting models}\ \ \subset\ \ \mbox{\it modular models}\ \ \subset\ \ \mbox{\it JEM}. \]
Even the smallest class, basic models, is already sufficient for mathematical completeness of justification logics. So, the main idea of progressing to Fitting models, modular models, or 
\emph{JEM}s is not a pursuit of completeness but rather a desire to offer natural and manageable models for a variety of epistemic situations involving evidence, belief, and knowledge. 

\subsection{Novelties}

Several things in this paper might be considered novelties.

1. Separating {\bf accepted justifications} and {\bf knowledge-producing justifications}, and making knowledge and justified belief derived notions. Building an adequate model for Russell's Prime Minister example as a showcase.

2. Streamlined exposition of generic semantics for justifications as sets of formulas, connecting basic justification models with the canonical model of propositional classical logic.

3. Asserting logic ${\sf J}^-$ instead of {\sf J} (cf. Section~\ref{other}) as in \cite{Art08}, to be the base system of justification epistemic logic. 

\section{Preliminaries}
Standard modal epistemic models have ``propositional'' precision, i.e., they do not distinguish sentences with the same truth values at each possible world. The expressive power of such models for analysis of justification, belief, and knowledge is rather limited, and so we have to ``go hyperintensional."\footnote{From \cite{Cre75}: ``Hyperintensional contexts are simply contexts which do not respect logical equivalence."} 
Specifically, if, at all possible worlds, $t$ is a justification for $F$ 
\[ \forces\lc{t}F ,\]
and $G$ has the same truth value as $F$ 
\[ \forces F\leftrightarrow G,\] 
we still cannot conclude that $t$ is a justification for $G$
\[ \not\forces\lc{t}G .\]
A natural example comes from mathematics: both statements $0=0$ and {\it Fermat's Last Theorem}, FLT, are true (proven) mathematical facts and hence are true at all possible worlds. However, we cannot claim that a proof of $0=0$ is a proof of FLT as well.

A sample justification logic analysis of standard epistemic situations (Gettier examples, Red Barn example) is presented in \cite{Art08} using Fitting models \cite{Fit05}
though, due to the relative simplicity of those examples, this analysis could be replicated in a bi-modal language (cf. \cite{Wil15}). 

However, we cannot go much farther without adopting a justification framework: the situation changes when we have to represent several conflicting pieces of evidence for a stated fact, cf. the following Russell example of 1912 (\cite{Rus12}): 

\begin{quote}
{\it 
If a man believes that the late Prime Minister's last name began with a `B,' he believes what is true, since the late Prime Minister was Sir Henry Campbell Bannerman\footnote{which was true in 1912}. But if he believes that Mr. Balfour was the late Prime Minister, he will still believe that the late Prime Minister's last
name began with a `B,' yet this belief, though true, would not be thought to constitute knowledge.
}
\end{quote}
To keep it simple, we consider proposition $B$ 
\[ \mbox{\it the late Prime Minister's last
name began with a `B,' }
\] 
with two justifications for $B$: the right one $r$ and the wrong one $w$; the agent chooses $w$ as a reason to believe that $B$ holds. 

\begin{quote}
There is a mathematical version of the story with a true proposition and its two justifications; one is correct, the other is not. Consider the picture\footnote{which the author saw on the door of the Mathematics Support Center at Cornell this February.}:
\begin{equation}\label{frac}
\frac{1\!\!\!\not 6}{\not{6}4}=\frac{1}{4} .
\end{equation}
The true proposition is ``$16/64 = 1/4$," the right justification is dividing both the numerator and the denominator by $16$, and the wrong (but shorter and easier) justification is simplifying as in~(\ref{frac}).
\end{quote}

To formalize Russell's scenario in modal logic (cf. \cite{Wil15}), we introduce two modalities: {\bf K} for knowledge and {\bf J} for justified belief. In the real world, 
\begin{itemize}
\item $B$ holds;
\item ${\bf J}B$ holds, since the agent has a justification $w$ for $B$;
\item ${\bf K}B$ does not hold; 
\end{itemize}
thus yielding the set of assumptions
\[ \Gamma=\{B, \ {\bf J}B, \ \neg{\bf K}B\}.\]

However, $\Gamma$ doesn't do justice to Russell's scenario: the right justification $r$ is not represented and $\Gamma$ rather corresponds to the same scenario but lacking $r$. The intrinsic epistemic structure of the example is not respected. 

Within the \emph{JEM} framework, introduced in this paper, we provide a model for Russell's Prime Minister example which, we wish to think, fairly represents its intrinsic epistemic structure.

\section{Generic logical semantics of justifications}

What kinds of logical objects are justifications?  When asked in a mathematical context ``what is a predicate?," we have a ready answer: a subset of a Cartesian product of the domain set.  Within an exact mathematical theory, there should be a similar kind of answer to the question ``what is a justification?''. 
 
We consider this question in its full generality which, surprisingly, yields a clean and meaningful answer. We assume the language of justification logic consists of two disjoint sets of syntactic objects: 
\begin{enumerate}
\item a set of {\bf justification terms} \emph{Tm}; 
\item a set of {\bf formulas} \emph{Fm}, built inductively from propositional atoms using Boolean connectives and the justification formula formation rule: if $F$ is a formula, 
$F\in\emph{Fm}$, and $t$ a justification term, $t\in\emph{Tm}$, then $\lc{t}F$ is again a formula, $\lc{t}F\in\emph{Fm}$.
\end{enumerate}

The meaning assigned to formulas is a classical truth value, $0$ for \emph{false} and $1$ for \emph{true}, and we retain classical logic behavior for propositional connectives. The key item is to give meaning to justification terms, and this will be a \emph{set of formulas} interpreted as \emph{the set of formulas for which it is a justification}. Here is a formal definition.

\begin{definition}[Basic Model]
A \emph{basic model}, simply called $\ast$, consists of an interpretation of the members of \emph{Fm}, and an interpretation of the members of \emph{Tm}.

The interpretation of a formula in a basic model is a truth value.  That is,
\[ \ast: \emph{Fm} \mapsto \{0,1\}. \]
We assume $(X\imp Y)^\ast = 1$ if and only if $X^\ast=0$ or $Y^\ast = 1$, and similarly for the other Boolean connectives. Let also $\models_ \ast X$ stand for $X^\ast =1$.

We interpret justification terms as \emph{sets of formulas}.  That is, 
\[  \ast: \emph{Tm}\mapsto 2^\emph{Fm}. \]
Our final requirement connects the two mapping roles that $\ast$ plays in a basic model.  For any $X\in\emph{Fm}$ and any $t\in\emph{Tm}$,
\[
\models_\ast \lc{t}X\ \mbox{ if and only if }\ X\in t^\ast.
\]
\end{definition}

It is easy to check that any mapping $\ast$ from propositional letters to truth values, and from justification terms to sets of formulas determines a unique basic model. 

So far, a basic model is nothing but a classical propositional model in which justification assertions $\lc{t}F$ are treated as \emph{independent propositional atoms}. 

Note that while propositions are interpreted semantically, as truth values, justifications are interpreted syntactically, as sets of formulas. This is a principal {\it hyperintensional} feature: a basic model may treat distinct formulas $F$ and $G$ as equal, i.e. $F^\ast=G^\ast$,  but still be able to distinguish justification assertions $\lc{t}F$ and $\lc{t}G$, e.g., when $F\in t^\ast$, but $G\not\in t^\ast$ yielding $\models_\ast \lc{t}F$ but $\not\models_\ast\lc{t}G$.

\begin{definition}\label{classicalcl}
Let $S\subseteq\emph{Fm}$ and $X\in\emph{Fm}$.  We write $S\proves X$ if formula $X$ is derivable from the set $S$ of formulas in classical logic treating justification assertions $\lc{t}F$ as propositional atoms, and with \emph{Modus Ponens} as the only rule of inference.  We say that $S$ is \emph{consistent} if $S\not\proves \bot$.
\end{definition}
{\bf A basic model of $S$ is merely a possible world containing $S$ in the canonical model}, i.e., a maximal consistent set $\Gamma$ of formulas, with the convenience agreement to read $\lc{t}F\in\Gamma$ as $F\in \{X \mid \lc{t}X\in\Gamma\}$. In this respect, basic models and the canonical model are slightly different but obviously equivalent ways of presenting the same object. When we move to more sophisticated models (Fitting models, modular models), the advantage of dealing with sets and operations (e.g. basic models) over logical conditions (e.g. the canonical model) becomes clear.

\begin{definition}
For $S\subseteq\emph{Fm}$, $\emph{BM}(S)$ is the \emph{class of all basic models of $S$}.
\end{definition}

\begin{Theor}[Generic Soundness and Completeness]\label{generic} Each set of formulas $S$ is sound and complete with respect to its class of basic models ${\it BM}(S)$. In other words, $S\proves F$ if and only if $F$ is true in each basic model of $S$.
\end{Theor}
\begin{proof} This theorem is merely a reformulation of the soundness and completeness of classical propositional logic with hypotheses. Indeed, if $S\proves F$ and $\models_\ast S$, then $\models_\ast F$ since propositional derivations respect validity. 

If $S\not\proves F$, then there is a Boolean evaluation $\ast$ which makes all formulas from $S$ true, $S^\ast=1$, and $F$ false, $F^\ast=0$. In this case there are two types of atomic propositions: propositional letters $P$ and justification assertions $\lc{t}X$. Define 
\[ t^\ast = \{ X\mid (\lc{t}X)^\ast=1\}\]
and note that $(\lc{t}X)^\ast=1$ iff $X\in t^\ast$. Therefore, $\ast$ is a propositional evaluation and $\ast$ is a basic model yielding the same truth values of atomic formulas $P$ and $\lc{t}X$. Since $S^\ast=1$ and $F^\ast=0$, we have $\models_\ast S$ and $\not\models_\ast F$ for basic model $\ast$.
\end{proof}

An easy corollary: $\proves F$ iff $F$ is a tautology (with $\lc{t}X$'s as distinct propositional atoms).  

\begin{example}\label{e1} We work with classical logic over \emph{Fm}, as in Definition~\ref{classicalcl}, taking $S=\emptyset$.
\begin{enumerate}
\item For any justification term $t$, \[\not\proves \lc{t}F.\] 
Straightforward, since $\lc{t}F$ is not a propositional tautology. For a specific countermodel put $t^\ast=\emptyset$ for each term $t\in\emph{Tm}$, which makes $\not\models_\ast\lc{t}F$. 

\item For any propositional letter $P$, and term $t$, 
\[\not\proves \lc{t}P\imp P.\] 
Likewise, this holds because $\lc{t}P\imp P$ is not a propositional tautology. Specifically, put $t^\ast=\emph{Fm}$ and $P^\ast=0$, with other  assignments being arbitrary. In this model, all justification assertions are true, but $\lc{t}P\imp P$ is false.

\item  For any propositional letter $P$, and term $t$, 
\[\not\proves P\imp \lc{t}P.\] 
Again, this holds since $P\imp \lc{t}P$ is not a propositional tautology. For example, put $t^\ast=\emptyset$ and $P^\ast=1$. In this model, $t$ is not a justification for $P$ (i.e., $\not\models_\ast \lc{t}P$) and $P\imp \lc{t}P$ is false. 

\item A somewhat less trivial example illustrating hyperintensionality: for a justification variable $x$ and formula $F$
\[\not\proves \lc{x}F\imp\lc{x}(F\land F).\]
The high-level argument is the same: formulas $\lc{x}F$ and $\lc{x}(F\land F)$ from Boolean evaluation point of view can be regarded as distinct propositional variables. Hence $\lc{x}F\imp\lc{x}(F\land F)$ is not a tautology. For a countermodel, take $x^\ast=\{F\}$. Then 
$\models_\ast\lc{x}F$, but $\not\models_\ast \lc{x}(F\land F)$. This 
demonstrates hyperintensionality of justification logic base, since $F$ and $F\land F$ are provably equivalent, but not $\lc{x}F$ and $\lc{x}(F\land F)$.
\end{enumerate} 
\end{example}

\section{Basic Justification Logic ${\sf J}^-$}
Within the Justification Logic framework, there are two sorts of logical objects: justification terms \emph{Tm} and formulas \emph{Fm}. Let us become more specific about both. 
\begin{itemize}
\item For \emph{Tm}, reserve a set of justification constants $a,b,c,\ldots$ with indices, and variables  $x,y,z,\ldots$ with indices. Justification terms are built from constants and variables by a binary operation $\cdot$  (application). 
\item Formulas are built from propositional letters $p,q,r,\ldots$ (with indices) and Boolean constant $\bot$ (falsum) by the standard Boolean connectives $\wedge,\vee,\imp,\neg$  with a new formation rule: {\it whenever $t$ is a justification term and $F$ is a formula, $\lc{t}F$ is a formula (with the informal reading ``$t$ is a justification for $F$'')}. 
\end{itemize}
\par\noindent
The \emph{logical system ${\sf J}^-$} consists of two groups of postulates.
\par
\begin{itemize}
\item {\bf Background logic}: axioms of classical propositional logic, rule \emph{Modus Ponens}.
\item {\bf Application}: $\lc{s}(F\imp G)\imp(\lc{t}F\imp\lc{[s\!\cdot\! t]}G)$. 
\end{itemize}
Basic models corresponding to ${\sf J}^-$ are the ones in which the application axiom holds. They can be specified by a natural combinatorial condition. 

\begin{definition}
For sets of formulas $S$ and $T$, we define \[S\triangleright T = \{ F \mid G\imp F\in S \ \mbox{and}\ G\in T \  \mbox{for some $G$} \}.\]
Informally, $S\triangleright T$ is the result of applying {\em Modus Ponens} once to all members of $S$ and of $T$ (in a given order). 
\end{definition}
\begin{Theor}\label{bmodelsforJ0} ${\it BM}({\sf J}^-)$ is the class of basic models with the following closure condition
\begin{equation}\label{ccapplication}
s^\ast\triangleright t^\ast\subseteq (s\!\cdot\! t)^\ast.
\end{equation}
\end{Theor}
\begin{proof} Let us assume the closure condition (\ref{ccapplication}) and check the validity of the application axiom. Indeed, 
$\models_\ast \lc{s}(F\imp G)$ and $\models_\ast\lc{t}F$ yield $(F\imp G)\in s^\ast$ and $F\in t^\ast$. By the closure condition, $G\in [s\!\cdot\! t]^\ast$, i.e., $\models_\ast \lc{[s\!\cdot\! t]}G$. 

Now assume the application axiom and derive the closure condition (\ref{ccapplication}). Let $(F\imp G)\in s^\ast$ and $F\in t^\ast$. By definitions, this yields $\models_\ast \lc{s}(F\imp G)$ and $\models_\ast\lc{t}F$. By the application axiom, $\models_\ast \lc{[s\!\cdot\! t]}G$, hence $G\in [s\!\cdot\! t]^\ast$. \end{proof} 

\begin{example}\label{e2} None of the formulas from Example~\ref{e1}:  $\lc{t}F$, $\lc{t}P\imp P$, $P\imp \lc{t}P$, $\lc{x}F\imp\lc{x}(F\land F)$ is derivable in ${\sf J}^-$. Indeed, every specific evaluation from Example~\ref{e1}.1--3 satisfies the closure condition (\ref{ccapplication}), hence their countermodels are ${\sf J}^-$-models. Consider the latter formula (4). Put $x^\ast=\{F\}$ and $t^\ast={\it Fm}$ for all other terms $t$. The closure condition (\ref{ccapplication}) holds vacuously, hence $\ast$ is a ${\sf J}^-$-model. Obviously, $\models_\ast \lc{x}F$ and $\not\models_\ast\lc{x}(F\land F)$.
\end{example}

Constants in justification logic are used to denote justifications of assumptions, in particular, axioms. Indeed, as we have already seen in Example~\ref{e2}, no formula $\lc{t}F$ is derivable in ${\sf J}^-$. In particular, no logical axiom is assumed justified in ${\sf J}^-$ which is not realistic. 

\begin{definition} A set $S$ of formulas is \emph{reflexive} if for each $\lc{s}\lc{t}F\in S$, $\lc{t}F$ is also in $S$. 
By constant specification \emph{CS} we understand a reflexive set of formulas of the type 
\[ \lc{c_n}\lc{c_{n-1}}\lc{c_{n-2}}\ \ldots\ \lc{c_1}A \]
where $A$ is a ${\sf J}^-$-axiom and $c_i$ are justification constants. The major classes of constant specifications: 
\begin{itemize}
\item {\it empty}, 
\item {\it total|}  (each constant is a justification for each axiom), 
\item {\it axiomatically appropriate} (each axiom has a justification at any depth),
\item {\it injective} (each constant justifies at most one axiom).
\end{itemize}
\end{definition}

Let  \emph{CS} be a constant specification. Then by ${\sf J}^-(\emph{CS})$, we understand ${\sf J}^-$ with additional axioms \emph{CS}. 
A \emph{CS-model} is a model in which all formulas from \emph{CS} hold. 

\begin{Coro}\label{corJ} Basic models for ${\sf J}^-({\it CS})$ are the basic CS-models for ${\sf J}^-$. ${\sf J}^-({\it CS})$ is sound and complete with respect to the class of its basic models.
\end{Coro}

\subsection{Sharp models}

In closure condition~(\ref{ccapplication}) from Theorem~\ref{bmodelsforJ0}, one cannot, generally speaking, replace the inclusion ``$\subseteq$" by the equality ``$=$" without violating completeness Theorem~\ref{generic}. 

Indeed, fix a justification constant $0$ and consider logic \[{\cal L} = {\sf J}^- + \{\neg\lc{0}F\mid F\in{\it Fm}\}.\]
Informally, justification $0$ receives empty evaluation in any basic model, $0^\ast=\emptyset$. We claim that formula 
$F= \neg\lc{[0\!\cdot\! 0]}P$
is not derivable in $\cal L$, but is true in any basic model of $\cal L$ with the closure condition 
$s^\ast\triangleright t^\ast = (s\!\cdot\! t)^\ast$. 
To show that ${\cal L}\not\proves F$, it suffices to find a basic model for $\cal L$ in which $F$ is false. Consider a basic model $\sharp$ such that $0^\sharp=\emptyset$ and $t^\sharp=\emph{Fm}$ for any other justification term $t$. Obviously, the closure condition from Theorem~\ref{bmodelsforJ0}, together with $0^\sharp=\emptyset$, is met. Therefore, $\sharp$ is a basic model of $\cal L$. It is immediate that $F$ is false in $\sharp$, since $[0\!\cdot\! 0]^\sharp=\emph{Fm}$.
On the other hand, for any basic model of $\cal L$ with the closure condition $(0\!\cdot\! 0)^\ast = 0^\ast\triangleright 0^\ast$, $(0\!\cdot\! 0)^\ast = \emptyset$ since 
$0^\ast = \emptyset$ and $\emptyset\triangleright \emptyset = \emptyset$. 

\begin{definition} \emph{Sharp} basic models are those in which the application closure condition has the form 
\begin{equation}\label{=application}
(s\!\cdot\! t)^\ast = s^\ast\triangleright t^\ast.
\end{equation}
\end{definition} 

One can prove that ${\sf J}^-$ and all its constant specification extensions of ${\sf J}^-(CS)$ are complete with respect to their sharp basic models. We skip the proof as not essential for the purposes of this paper.  

Note that a sharp model is completely defined by evaluations of atomic propositions and atomic justifications. 

\subsection{Other justification logics}\label{other}

There is a whole family of justification logics and they all extend ${\sf J}^-$; the reader is referred to \cite{Art08,AF15} for details. Here we list just the main systems of justification logic for a general orientation. 

Logic {\sf J} is obtained from ${\sf J}^-$ by adding a new operation on justifications `$+$' and the Sum principle
\[ \lc{s}F\vee\lc{t}F\imp\lc{[s+t]}F .\]
Logics {\sf JD}, {\sf JT}, {\sf J4}, {\sf J5}, etc., are obtained by adding the corresponding combination of principles
\[ D = \neg\lc{t}\bot ,\]
\[ T= \lc{t}F\imp F , \]
\[ 4 = \lc{t}F\imp\ \lc{!t}\lc{t}F , \]
\[ 5 = \neg\lc{t}F\imp\ \lc{?t}\neg\lc{t}F .\]

Historically, the first justification logic was the Logic of Proofs {\sf LP} \cite{Art95TR,Art01a} which in these notations can be identified as {\sf JT4}. 
By now, the family of justification logics has grown to be infinite, cf. \cite{Fit16}.

\section{Justification Epistemic Models}\label{JEMs}

We need more expressive power to capture epistemic differences between justifications and their use by the knower. Some justifications are knowledge-producing, some are not. The agent makes choices on which justifications to base an agent's beliefs/knowledge and which justifications to ignore in this respect. These actions are present in epistemic scenarios, out of which we will primarily focus on Russell's example, which has it all:
\begin{itemize}
\item there are justifications $w$ (Balfour was the late prime minister) and $r$ (Bannerman was the late prime minister) for $B$; 
\item $r$ is knowledge-producing whereas $w$ is not;
\item the agent opts to base his belief on $w$ and ignores $r$;
\item the resulting belief is evidence-based, but is not knowledge. 
\end{itemize}

Fix ${\sf J}^-(CS)$ for some axiomatically appropriate constant specification \emph{CS}. 

\begin{definition}  A set $X$ of justification terms is {\it properly closed} if $X$ contains all constants and is closed under applications.
If $Y$ is a set of justification terms, then by $\overline{Y}$ we mean the proper closure of $Y$, i.e., the minimal properly closed superset of $Y$.  
\end{definition}

\begin{definition}\label{LEMdef}
A (basic) {\it Justification Epistemic Model} (JEM) is $$(\ast, {\cal A}, {\cal E})$$ where 
\begin{itemize}
\item $\ast$ is a basic ${\sf J}^-(CS)$-model;
\item ${\cal A}\subseteq{\it Tm}$ is a properly closed set ${\cal A}$ of {\it accepted justifications}\footnote{Consistency $\models_\ast\neg\lc{t}\bot$ for any $t\in{\cal A}$ may or may not be  required.};
\item ${\cal E}\subseteq{\it Tm}$ is a properly closed sets ${\cal E}$ of {\it knowledge-producing justifications}; all knowledge-producing justification should be factive, $\models_\ast\lc{t}F\imp F$ for each $F$ and each $t\in{\cal E}$.
\end{itemize}
\end{definition}
Both sets ${\cal A}$ and ${\cal E}$ contain all constants. This definition presumes that constants in a model are knowledge-producing and accepted. 

\begin{definition}
Sentence $F$ is {\it believed} if there is $t\in{\cal A}$ such that $\models_\ast\lc{t}F$. Sentence $F$ is {\it known}  if there is $t\in{\cal A}\cap{\cal E}$ such that $\models_\ast\lc{t}F$.
\end{definition}

By \emph{ground term} we understand a term containing no (justification) variables. In other words, a term is ground iff it is built from justification constants only. 

Sets of accepted and knowledge-producing justifications overlap on ground terms but otherwise are in a general position\footnote{In principle, one could consider smaller sets $\cal A$, which would correspond to the high level of skepticism of an agent who does not necessarily accept logical truths (axioms) as justified. We leave this possibility for further studies.}. There may be accepted but not knowledge-producing justifications, and vice versa. So, \emph{JEM}s do not analyze {\bf why} certain justifications are knowledge-producing (knowledge-grade) or accepted, but rather provide a formal framework that accommodates these notions. 

\subsection{Injective justifications}
The notions of {\it accepted} and {\it knowledge-producing} justifications should be utilized with some caution. 
Imagine a justification $t$ for $F$ (i.e., $\lc{t}F$ holds) and for $G$ ($\lc{t}G$) such that, intuitively, $t$ is a knowledge-producing justification for $F$ but not for $G$. Is such a $t$ knowledge-producing, trustworthy, acceptable for a reasonable agent? The answers to these questions seem to depend on $F$ and $G$, and if we prefer to handle justifications as objects rather than as justification assertions, it is technically convenient to assume that justifications are {\it injective}, or statement-specific: 
\[ \mbox{\it there is at most one formula $F$ such that $\lc{t}F$ holds}. \]
Conceptually, by going injective with justifications, one does not lose generality: if $p$ is a proof of $F$ and of something else, then the same $p$ with a designated statement $F$, symbolically, a pair $(p,F)$, can be regarded as an injective/single conclusion proof of $F$. So, with injective justifications without losing generality we gain the ability to go back and forth between justifications and statements. 

In model $\cal R$ for the Russell Example, Section~\ref{Russell}, all justifications are injective. 

Note that ${\sf J}^-$ is not complete with respect to the class of basic models which are both sharp and injective (as model $\cal R$ for the Russell Example). Indeed, consider formula $F$, 
\[ F = \neg(\lc{x}(P\imp Q)\land \lc{y}P\land \lc{[x\!\cdot\! y]}R)\]
where $P,Q,R$ are distinct propositional letters and $x,y$ justification variables. Obviously, $F$ holds in any basic model $\ast$ which is sharp and injective. Imagine a sharp injective 
$\ast$ in which $\lc{x}(P\imp Q)$ and $\lc{y}P$ hold. In such $\ast$, $[x\!\cdot\! y]^\ast=\{Q\}$, hence both $\neg\lc{[x\!\cdot\! y]}R$, and $F$ hold. On the other hand, $F$ is not derivable in 
${\sf J}^-$, e.g., $F$ fails in the basic model $\ast$ with $x^\ast=\{P\imp Q\}$, $y^\ast=\{P\}$, and $t^\ast={\it Fm}$ for any other $t$ (check closure condition~(\ref{ccapplication})!). So, ``sharp and injective" justification tautologies constitute a proper extension of ${\sf J}^-$. It would be interesting to find its complete axiomatization. It would seem that functionality axioms from \cite{AS93,Kru97,Kru02} might be relevant here.


\section{Russell scenario as a {\it JEM}}\label{Russell}
Consider the version of ${\sf J}^-$ in language with two justification variables $w$ and $r$, one propositional letter $B$, and injective constant specification \emph{CS}:
\[ \lc{c_n}A\in{\it CS}\ \ \mbox{iff}\ \ \mbox{\it $A$ is an axiom and $n$ is the G\"odel number of $A$.} \]
Define a model $\ast$ such that 
\begin{itemize}
 \item $B^\ast = 1$, i.e., $\models_\ast B$;
 \item $c_n^\ast=\{A\}$ if $A$ is an axiom and $n$ is the G\"odel number $|A|$ of $A$, and $c_n^\ast=\emptyset$ otherwise;
 \item $w^\ast=r^\ast=\{B\}$, e.g.,  $\models_\ast \lc{r}B$ and $\not\models_\ast\lc{r}F$ for any $F$ other than $B$ (the same for $w$);
 \item Application is sharp: $(s\!\cdot\! t)^\ast = s^\ast\triangleright t^\ast$.
 \end{itemize}
A \emph{JEM} $\cal R$ (for Russell's scenario) is $(\ast,{\cal A},{\cal E})$ with 
\begin{itemize}
\item ${\cal A}=\overline{\{w\}}$, i.e., the set of accepted justifications is $\{w\}$, properly closed;
\item ${\cal E}=\overline{\{r\}}$, i.e., the set of knowledge-producing justifications is $\{r\}$, properly closed.
\end{itemize}

Though the idea behind $\cal R$ is quite intuitive, we need to fill in some technical details: extending truth evaluations to all terms and formulas and checking closure conditions. 

\subsection{Technicalities of the model}\label{technicalities}

Put $c_{|A|}^\ast=\{A\}$ for each axiom $A$ of ${\sf J}^-(CS)$. Technically, this is an inductive definition. Base: $c_0^\ast = \emptyset$, given $0$ is not a G\"odel number of any formula. Inductive step: suppose $n$ is the G\"odel number $|F|$ of some formula $F$. 
If $F$ is an axiom of $J^-$, put $c_n^\ast=\{F\}$. If $F=\lc{c_k}G$ for some $c_k$ and $G$, then, by monotonicity of G\"odel numbering, 
$k<n$, hence $c_k^\ast$ is defined. If $c_k^\ast = \{G\}$, then $\lc{c_k}G$ is an axiom of ${\sf J}^-(CS)$ and we can put $c_n^\ast=\{F\}$.
In all other cases, $c_n^\ast = \emptyset$. 

Since application is sharp, the evaluation of each term is, at most, a singleton. Together with Boolean truth tables, this determines the truth value of any formula. 

\begin{Lemma}\label{fact} Each $t\in {\it Tm}$ is factive, $\models_\ast \lc{t}F\imp F$. 
\end{Lemma} 
\begin{proof} Induction on $t$. 
Assume $\models_\ast\lc{t}F$; that means $t^\ast=\{F\}$. If $t$ is $w$ or $r$, then $F$ is $B$, which is true in the model $\ast$.  If $t$ is a constant, then $F$ is an axiom and hence true in $\ast$. The induction step corresponds to application, which preserves the truth of justified formulas. 
\end{proof}

From Lemma~\ref{fact} it follows that all accepted justifications are consistent and all knowledge-producing justifications are factive. Therefore, ${\cal R}=(\ast, {\cal A}, {\cal E})$ is indeed a \emph{JEM}. 

\subsection{Main result}

\begin{Theor}\label{main} In model $\cal R$, sentence $B$ is true, justified and believed, but not known.  
\end{Theor}
\begin{proof}
In model $\cal R$, sentence $B$ is 
\begin{itemize}
\item true, since $\models_\ast B$; 
\item justified, since $\models_\ast \lc{w}B$;
\item believed, since $w\in{\cal A}$.
\end{itemize}
We have to show that for any justification $g\in{\cal A}\cap{\cal E}$, $\not\models_\ast \lc{g}B$. 

Consider an auxiliary basic model $\bullet$ which is the same as $\ast$ but with $B^\bullet=0$, i.e., the truth value of $B$ is flipped from `true' to `false.'  In particular, application in $\bullet$ is sharp.

\begin{Lemma}\label{asts} For each justification term $t$, $$t^\ast = t^\bullet.$$ 
\end{Lemma}
\begin{proof} The inductive process (based on sharp application) of evaluating all justifications given evaluations of atomic justifications operates only with formulas of type $\lc{t}F$ and starts with the same initial set of such formulas in $\ast$ and $\bullet$. Hence the results of these processes in  $\ast$ and $\bullet$ coincide. 

\end{proof}
In particular, for all $g\in{\cal A}\cap{\cal E}$, $g^\ast = g^{\bullet}$, and if $\models_\ast \lc{g}B$, then $\models_\bullet \lc{g}B$ as well. 

\begin{Lemma}\label{gfact} Each $g\in{\cal A}\cap{\cal E}$ is factive in $\bullet$, i.e., $\models_\bullet\lc{g}F\imps F$.
\end{Lemma} 
\begin{proof} All  $g\in{\cal A}\cap{\cal E}$ are obtained from constants by application. By construction, if $\models_\bullet\lc{c}X$, then 
$X\in c^\ast$ and $X$ is an axiom, hence true.  Application obviously preserves factivity. 

\end{proof}
To complete the proof of Theorem~\ref{main}, suppose ${\cal R}\models \lc{g}B$, i.e., $\models_\ast \lc{g}B$, for some $g\in{\cal A}\cap{\cal E}$. By Lemma~\ref{asts}, $\models_\bullet \lc{g}B$, and, by Lemma~\ref{fact}, $\models_\bullet B$, which is not the case. 

\end{proof}

\subsection{Can Russell's scenario be made modal?}
One could try to express Russell's scenario in a modal language by introducing the justified belief modality
\[ {\bf J}F\ \ \Leftrightarrow\ \ \mbox{\it there is $t\in{\cal A}$ such that $\models\lc{t}F$}, \]
and the knowledge-producing modality 
\[ {\bf E}F\ \ \Leftrightarrow\ \ \mbox{\it there is $t\in{\cal E}$ such that $\models\lc{t}F$}, \]
and by stipulating that $F$ is known iff $F$ is both accepted and supported by a knowledge-producing justification:
\[ {\bf K}F\ \ \Leftrightarrow\ \ {\bf J}F\wedge{\bf E}F.\]

This, however, fails, since both ${\bf J}B$ and ${\bf E}B$ hold in $\cal R$, but ${\bf K}B$ does not. We are facing a Gettier-style phenomenon, when a proposition is supported by a knowledge-producing justification (hence true), and believed, but not known (since knowledge-producing and accepted justifications for $B$ are different). This once again illustrates the limitations of modal language in tracking and sorting justifications.

\section{Multi-world {\it JEM}s and epistemic Kripke models}\label{toconventional}

A multi-world version of {\it JEM} is 
\[{\cal M}=(W,\ast,{\cal A},{\cal E})\]
such that $W$ is a nonempty set of states (possible worlds) and for each $u\in W$, the corresponding $(\ast,{\cal A},{\cal E})$ is a single-state {\it JEM}. 

Let us outline what simplifying assumptions would make a {\it JEM} ${\cal M}=(\ast,{\cal A},{\cal E})$ a Kripke model for modal logic. 

First, we have to convert justification assertions $\lc{t}F$ into modal epistemic assertions $\Box F$. A natural way of accomplishing this would be postulating a {\it justification indifference} property: 
\[ \lc{t}F\ \leftrightarrow\ \lc{s}F \]
holds at each world.  

Second, we have to adopt an additional {\em fully explanatory property} of the model. We define the accessibility relation $R$ on $W$ as 
\[ uRv\ \ \ \ \mbox{\it iff}\ \ \ \ \mbox{\it for all $t$ and $F$,}\ u\models \lc{t}F \Rightarrow v\models F .\]
This standard definition secures that 
\[ u\models \lc{t}F\ \ \Rightarrow\ \ R(u)\models F , \]
but not ``$\Leftrightarrow$,'' which is a hidden assumption in Kripke models, cf. \cite{Art16KTM}. We have to make this assumption explicit. 
A model ${\cal M}$ is {\em fully explanatory}, if 
\[ R(u)\models F\ \ \Rightarrow\ \ \mbox{\it for some $t$}, u\models \lc{t}F, \]
cf.  \cite{Fit05} in which this property was introduced and studied. The fully explanatory property can be viewed as a propositional form of knowledge of the model, cf. \cite{Art16KTM}.

It is now straightforward that modal logic {\sf T} corresponds to {\it JEM}s with justification indifference and the fully explanatory property (a narrow special case of {\it JEM}s). To obtain more conventional {\sf S5}-models, we have to adopt yet another additional assumption of {\it decidability of justifications}, but leave this to further studies.  


\section{Discussion}
Model $\cal R$ fairly represents a situation with two justifications, right and wrong, for a given proposition and, in this respect, serves as a showcase for modeling epistemic situations with fallible justifications. 

There are many natural open questions that indicate possible research directions: 

\begin{itemize}
\item Are justification assertions checkable, decidable for an agent? 
\item Is the property of a justification to be knowledge-producing checkable by the agent? 
\item In multi-agent cases, how much do agents know about each other and about the model? 
\item Do agents know each other's accepted and knowledge-producing justifications? 
\item What is a complete axiomatization of sharp and injective models over ${\sf J}^-$? 
\item What is the complexity of these new justification logics and what are their feasible fragments which make sense for epistemic modeling? 
\end{itemize}

\section{Acknowledgements}
The author is grateful to Melvin Fitting, Vladimir Krupski, Elena Nogina, and Tudor Protopopescu for helpful suggestions. 
Special thanks to Karen Kletter for editing and proofreading this text.

\end{document}